\documentclass[11pt]{article}
\usepackage{amsmath}
\usepackage{amssymb}
\usepackage{amsthm}
\usepackage{amsfonts}
\usepackage{color}

\newcounter{AbcT}

\newtheorem {Theorem}    {Theorem}[section]

\newtheorem {Lemma}      [Theorem]    {Lemma}

\newtheorem {Proposition}[Theorem]    {Proposition}

\theoremstyle{definition}   

\newcommand{\ignore}[1]{}

\def\eps{\epsilon}

\def\P{{\bf{P}}}

\def\v0{{\bf 0}}

\def\0{\hat{0}}
\def\1{\hat{1}}

\def\path{{\tt path}}

\def\phi{\varphi}

\def\be{\begin{equation}}
\def\ee{\end{equation}}

\definecolor{Red}{rgb}{1,0,0}
\definecolor{Blue}{rgb}{0,0,1}
\definecolor{Olive}{rgb}{0.41,0.55,0.13}
\definecolor{Green}{rgb}{0,1,0}
\definecolor{MGreen}{rgb}{0,0.8,0}
\definecolor{DGreen}{rgb}{0,0.55,0}
\definecolor{Yellow}{rgb}{1,1,0}
\definecolor{Cyan}{rgb}{0,1,1}
\definecolor{Magenta}{rgb}{1,0,1}
\definecolor{Orange}{rgb}{1,.5,0}
\definecolor{Violet}{rgb}{.5,0,.5}
\definecolor{Purple}{rgb}{.75,0,.25}
\definecolor{Brown}{rgb}{.75,.5,.25}
\definecolor{Grey}{rgb}{.5,.5,.5}
\definecolor{Black}{rgb}{0,0,0}

\title{Arrow's Impossibility Without Unanimity}

\author{Elchanan Mossel
\thanks{Weizmann Institute and U.C. Berkeley. Supported by an Alfred Sloan fellowship
  in Mathematics, by NSF CAREER grant DMS-0548249 (CAREER), by DOD ONR grant (N0014-07-1-05-06), by BSF grant 2004105 and by ISF grant
  1300/08.}
  }
\begin{document}
\maketitle

\begin{abstract}
Arrow's Impossibility Theorem states that any constitution which satisfies Transitivity, Independence of Irrelevant Alternatives (IIA) and Unanimity is a dictatorship.
Wilson derived properties of constitutions satisfying Transitivity and IIA for unrestricted domains where ties are allowed.
In this paper we consider the case where only strict preferences are allowed. In this case
we derive a new short proof of Arrow theorem and further obtain a new and complete characterization of all functions satisfying Transitivity and IIA. The proof is based on a variant of the method of pivotal voters due to Barbera. 
\end{abstract}

\section{Introduction}
Arrow's Impossibility theorem~\cite{Arrow:50,Arrow:63} states that certain properties cannot hold simultaneously for constitutions on three or more alternatives. Consider $A = \{a,b,\ldots,\}$, a set of $k \geq 3$ alternatives. A (strict) {\em transitive preference} over $A$ is a ranking of the alternatives from top to bottom where ties are not allowed. Such a ranking corresponds to a {\em permutation} $\sigma$ of the elements $1,\ldots,k$ where $\sigma_i$ is the rank of alternative $i$.

We consider a society of $n$ individuals labeled $1,\ldots,n$, each of them has a transitive preference. A {\em constitution} is a function $F$ that associates to every $n$-tuple $\sigma = (\sigma(1),\ldots,\sigma(n))$ of transitive preferences (also called {\em profile}),
and every pair of alternatives $a,b$ a preference between $a$ and $b$.
Some basic properties of constitutions are:



\begin{itemize}
\item
{\em Transitivity}. The constitution $F$ is {\em transitive} if $F(\sigma)$ is transitive for all $\sigma$.
In other words, for all $\sigma$ and for all three alternatives $a,b$ and $c$, if $F(\sigma)$ prefers $a$ to $b$, and prefers $b$ to $c$, it also prefers $a$ to $c$.
\item
{\em Independence of Irrelevant Alternatives (IIA).} The constitution $F$ satisfies the IIA property if for every pair of alternatives $a$ and $b$, the social ranking of $a$ vs. $b$ (higher or lower) depends only on their relative rankings by all voters.
\item
{\em Unanimity.} The constitution $F$ satisfies {\em Unanimity} if the social outcome ranks $a$ above $b$ whenever all individuals rank $a$ above $b$.
\item
{\em Dictatorship}. The constitution $F$ is a {\em dictatorship by individual $i$} or the $i$'th {\em dictator} if the social outcome ranking of voter $i$ so that either $F(\sigma) = \sigma(i)$ for all $\sigma$ or $F(\sigma) = -\sigma(i)$ for all $\sigma$ where
$-\sigma(i)$ is the ranking $\sigma_k(i) > \sigma_{k-1}(i) \ldots \sigma_2(i) > \sigma_1(i)$ by reversing the ranking $\sigma$.
\item
{\em Non-Imposition (NI)}. The constitution $F$ is a {\em Non-Imposition (NI)} if every transitive
preference order is achievable by some profile of transitive preferences.
\item
{\em Weak Non Imposition (WNI)}.
The constitution $F$ is a {\em Weak Non-Imposition (WNI)} if for every pair of alternatives $a,b$,
there exists a profile where the constitution ranks $a$ above $b$.
\item
{\em Non Degeneracy (ND)}.
The constitution $F$ is a {\em Degenerate} if there exists an alternative $a$ such that for all profiles $F$ ranks at the top (bottom).
The constitution $F$ is {\em Non Degenerate (ND)} if it is not degenerate.
\end{itemize}
We note that
\[
\mbox{Unanimity } \implies \mbox { NI } \implies \mbox{ WNI } \implies \mbox{ND},
\]
and that non of the reverse implications hold.

Our definition of dictator is more general than the standard definition as dictatorship.
To see that this is needed in the setup considered here, look at the constitution on $3$ alternatives and a single voter that assigns to the ranking $\sigma_1 > \sigma_2 > \sigma_3$, the reverse ranking $\sigma_3 > \sigma_2 > \sigma_1$. This constitution satisfies Transitivity, IIA and WNI. However it is not a dictator according to the "standard definition".
Note furthermore that the only dictator function which satisfies the Unanimity condition is the identity map and therefore Theorem~\ref{thm:general} implies Theorem~\ref{thm:arrow}.

The basic idea behind Arrow's impossibility is that there is some "tension" between IIA and Transitivity.
The original statement of Arrow theorem~\cite{Arrow:50} states that if the constitution $F$ is monotone, satisfies Non-Imposition, IIA and Transitivity then $F$ has to be a dictator. Later Arrow derived a stronger statement~\cite{Arrow:63}
\begin{Theorem} \label{thm:arrow}
Any constitution on three or more alternatives which satisfies Transitivity, IIA and Unanimity is a dictatorship.
\end{Theorem}
The result above is also known to hold when individuals are allowed to have as their preferences all the strict preferences as well as all preferences with ties.

In our first result we relax the unanimity condition and show that
\begin{Theorem} \label{thm:general}
Any constitution on three or more alternatives which satisfies Transitivity, IIA and WNI is a dictatorship.
\end{Theorem}
A similar result was proven in a beautiful paper by Wilson~\cite{Wilson:72} in the case where the voters preferences include preferences with ties.

We briefly comment on the difference between the two results.

\begin{itemize}
\item
We note that in general one cannot conclude impossibility results for one domain of preferences given that the same impossibility result hold for a larger domain of preferences.
Therefore Wilson's result does not imply ours.

\item
A second comment relates to the conclusion of the theorem. While our theorem explicitly characterizes all $2n$ functions that satisfy Transitivity, IIA and WNI, Wilson's result only asserts that such a function is a weak dictator in the following sense. Whenever the dictator states a strict preference, this preference is followed by society. However if the dictator preference ties two alternatives, the social outcome may be a non-trivial functions of the other voters preferences. Wilson does not give a complete characterization of the functions satisfying Transitivity, IIA and WNI in this setup.

\item
Thirdly, the two proofs are quite different. In particular, the proof given here is based on a "local" argument involving 2 voters and 3 alternatives. The local nature of the proof allowed to use the proof technique developed here to derive a quantitative Arrow theorem in follow up work by the author (for more details see the conclusion section).

\item
We finally like to note that while Wilson's result does not imply the results proven here, it seems possible to adapt his proof strategy in order to obtain the same results. However, such a proof does not seem to be as helpful in extending the results to quantitative setups.
\end{itemize}

Next we consider similar results where the WNI condition is omitted.  In some scenarios, where some opinions should be given negative considerations, the Unanimity condition is not natural. The WNI condition is much more natural. If the WNI condition does not hold then the constitution $F$ always ranks $a$ above $b$ for some alternatives $a,b$.

Note that in this case we must allow to include any constitution on two alternatives. Similarly it should allow a constitution on $4$ alternatives $a,b,c,d$ such that
$a,b$ are always ranked above $c,d$,the $a$ vs. $b$, ranking is decided arbitrarily according to the individual $a$ vs. preferences and similarly for the $c$ vs. $d$ ranking.

For the characterization it would be useful to write $A >_F B$ for the statement that for all $\sigma$ it holds that $F(\sigma)$ ranks
all alternatives in $A$ above all alternatives in $B$. We will further write $F_{A}$ for the constitution $F$ restricted to the alternatives in $A$. The IIA condition implies that $F_A$ depends only on the individual rankings of the alternatives in the set $A$.

\begin{Theorem} \label{thm:very_general}
A constitution $F$ on $k \geq 2$ alternatives satisfies IIA and Transitivity if and only if the following hold.
There exist a partition of the set of alternatives into disjoint sets $A_1,\ldots,A_r$ such that:
\begin{itemize}
\item
\[
A_1 >_F A_2 >_F \ldots >_F A_r,
\]
\item
For all $A_s$ s.t. $|A_s| \geq 3$ there exists a voter $j$ such that $F_{A_s}$ is a dictator on voter $j$.
\item
For all $A_s$ such that $|A_s| = 2$, the constitution $F_{A_s}$ is an arbitrary non-constant function of the preferences on the alternatives in $A_3$.
\end{itemize}
\end{Theorem}
Note that for sets $s$ such that $|A_s| = 2$, the constitution $F_{A_s}$ is arbitrary (except for the constant functions).
Note furthermore that Theorem~\ref{thm:very_general} implies Theorem~\ref{thm:general} since the WNI condition implies that the partition into sets contains only one set - the set of all candidates. Again, we note that a similar result for preferences with ties is given by Wilson~\cite{Wilson:72}. As before we note that our result gives a full characterization while Wilson's only a partial characterization, Wilson's result do not imply ours and the proof is different.

\subsection{Acknowledgement} We thank Salvador Barbera for the reference to the work of Wilson~\cite{Wilson:72}.

\section{A Short Proof of Arrow's Theorem}
Our proof proceeds of Theorem~\ref{thm:general} proceeds in $3$ steps - the base case is $n=2$ voters and $k=3$ alternatives.
We then generalize to $k=3$ alternatives and any number of voters. Finally we prove the result for arbitrary number of voters and alternatives. We later prove Theorem~\ref{thm:very_general} which requires a more detailed understanding of which alternatives "interact".

\subsection{Preliminaries}
Recall that we denote the profile of $n$ rankings by $\sigma = (\sigma(1),\ldots,\sigma(n))$.
For each pair of alternatives $a,b$ we write $x^{a > b} = (x^{a>b}(1),\ldots,x^{a>b}(n)) $ for the vector whose $i$'th coordinate is $1$ if voter $i$ prefers $a$ to $b$ and $-1$ if voter $i$\ prefer $b$ to $a$.

\begin{Proposition}
The IIA assumption implies that for all $a,b$ there exists a function $f^{a>b} : \{-1,1\}^n \to \{-1,1\}$ such that the constitution ranks $a$ ahead of $b$ if $f^{a>b}(x^{a>b}) = 1$ and ranks $b$ ahead of $a$ if $f^{a>b}(x^{a>b}) = -1$.
\end{Proposition}

As in previous proofs~\cite{Barbera:80,Geanakoplos:05}
a key notion is that of {\em pivotal} voter. Recall that voter $i$ is {\em pivotal} for $f: \{-1,1\}^n \to \{-1,1\}$
if there exists $x,y \in \{-1,1\}^n$ such that $x_j = y_j$ for $j
\neq i$ and $x_i \neq y_i$ and $f(x) \neq f(y)$. We will need the following easy facts:
\begin{Proposition}
The following hold:
\begin{itemize}
\item
If $f: \{-1,1\}^n \to \{-1,1\}$ is not constant then $f$ has at least one pivotal voter.
\item
Assume that there exists a voter $1 \leq i \leq n$ such that for all pairs of alternatives $a,b$
it holds that either $f^{a>b}$ is constant or $i$ is the only pivotal voter for $f^{a>b}$ then $F(\sigma)$ is of the form
$F(\sigma) = G(\sigma(i))$ for some function $G$.
\end{itemize}
\end{Proposition}
\begin{proof}
The first assertion follows by noting that changing any single coordinate in $x$ will not change the value of $f$ and
therefore the same is true for any number of coordinates so $f$
has to be constant. The second assertion follows from the first by noting that fixing the ranking of voter $i$ results in a constant function of the other voters.
Therefore there exists a function $G$ of voter $i$ such that $F(\sigma) = G(\sigma(i))$.
\end{proof}

\subsection{Different Pivots for Different Choices imply Non-Transitivity}
We begin by considering the case of $3$ candidates named $a,b,c$ and two voters named $1$ and $2$.
We note that
\begin{Proposition}
For all $i$:
\[
\{(x_i^{a>b}(\sigma_i),x_i^{b>c}(\sigma_i),x_i^{c>a}(\sigma_i) : \sigma_i \in S_3 \} =
\{-1,1\}^3 \setminus \{(1,1,1),(-1,-1,-1)\}.
\]
Similarly, the outcome of the constitution given by $f^{a>b},f^{b>c}$ and $f^{c>a}$ is non-transitive if
and only if
\[
(f^{a>b}(x^{a>b}),f^{b>c}(x^{b>c}),f^{c>a}(x^{c>a})) \in \{(-1,-1,-1),(1,1,1)\}.
\]
\end{Proposition}

Our main technical tool is the following result which is a rediscovery of a result of Barbera~\cite{Barbera:80}. 
The proof in~\cite{Barbera:80} uses Arrow's logic relation notation. We give the proof using binary bits below. 
\begin{Theorem} \label{thm:tech}
Consider a social choice function on $3$ candidates $a,b$ and $c$ and $n$ voters denoted $1,2,\ldots,n$. Assume that the social choice function satisfies that IIA condition and that there exists voters $i \neq j$ such that voter $i$ is pivotal for $f^{a>b}$ and voter $j$ is pivotal for $f^{b>c}$.
Then there exists a profile for which
$(f^{a>b}(x^{a>b}),f^{b>c}(x^{b>c}),f^{c>a}(x^{c>a}))$ is non-transitive.
\end{Theorem}

\begin{proof}
Without loss of generality assume that voter $1$ is pivotal for $f^{a>b}$ and voter $2$ is pivotal for $f^{b>c}$.
Therefore there exist $x_2,\ldots,x_n$ satisfying
\begin{equation} \label{eq:xpiv}
f^{a>b}(+1,x_2,\ldots,x_n) \neq f^{a>b}(-1,x_2,\ldots,x_n)
\end{equation}
and $y_1,y_3,\ldots,y_n$ satisfying
\begin{equation} \label{eq:ypiv}
f^{b>c}(y_1,+1,y_3,\ldots,y_n) \neq f^{b>c}(y_1,-1,y_3,\ldots,y_n).
\end{equation}
Let $z_1 = -y_1$ and $z_i = -x_i$ for $i \geq 2$.
By~(\ref{eq:xpiv}) and~(\ref{eq:ypiv}) we may choose $x_1$ and $y_2$ so that
\[
f^{a>b}(x) = f^{b>c}(y) = f(z),
\]
where $x = (x_1,\ldots,x_n),y=(y_1,\ldots,y_n)$ and $z=(z_1,\ldots,z_n)$.
Note further, that by construction for all $i$ it holds that
\[
(x_i,y_i,z_i) \notin \{(1,1,1),(-1,-1,-1)\},
\]
and therefore there exists a profile $\sigma$ such that
\[
x = x(\sigma), \quad y = y(\sigma), \quad z = z(\sigma).
\]
The proof follows.
\end{proof}

\subsection{A Proof of Arrow Theorem}
\begin{proof}
Assume the constitution satisfies the IIA property. For each pair of candidate let $f^{a>b}$ be the choice function between alternatives $a$ and $b$ and let $P^{a>b}$ be the set of pivotal voters for $f^{a>b}$.

By the Unanimity assumption if follows that $f^{a>b}$ takes both values $1$ and $-1$ and therefore
$P^{a>b}$ is not empty.

Assuming that Transitivity holds,
Theorem \ref{thm:tech} implies that for all $a,b,c$ it holds that
$P^{a>b}=P^{b>c}=\{i\}$ for some $i$. Note that this implies that for all $a,b,c,d$ it holds that
$P^{a>b}=P^{b>c}=P^{c>d}=\{i\}$. In other words, there exist a voter $i$ such that for all $a,b$, the voter $i$ is the single pivotal voter for $f^{a>b}$.

This implies that for all $a,b$ either $f^{a>b}(x) = x_i$ or $f^{a>b} = -x_i$ and by Unanimity it must be
that $f^{a>b}(x) = x_i$ for all $a,b$. This implies that the constitution is a dictator on voter $i$ as needed.
\end{proof}

\subsection{$n$ voters, 3 Candidates}

In order to prove Theorem~\ref{thm:general} we need the following proposition regarding constitutions of a single voter.
\begin{Proposition} \label{prop:1}
Consider a constitution $F$ of a single voter and three alternatives $\{a,b,c\}$ which satisfies IIA and transitivity.
Then exactly one of the following conditions hold:
\begin{itemize}
\item
$F$ is constant. In other words, $F(\sigma) = \tau$ for all $\sigma$ and some fixed $\tau \in S(3)$.
\item
There exists an alternative $c$ such that $c$ is always ranked at the top (bottom) of the ranking and
$f^{a>b}(x) = x$ or $f^{a>b}(x) = -x$.
\item
$F(\sigma) = \sigma$ for all $\sigma$
\item
$F(\sigma) = -\sigma$ for all $\sigma$.
\end{itemize}
\end{Proposition}

\begin{proof}
Assume $F$ is not constant, then there exist two alternatives $a,b$ such that $f^{a>b}$ is not constant and therefore $f^{a>b}(x) = x$ or $f^{a>b}(x)=-x$.
Let $c$ be the remaining alternative. If $c$ is always ranked at the bottom or the top the claim follows.
Otherwise one of the functions $f^{a>c}$ or $f^{b>c}$ is not constant. We claim that in this case all three functions are non-constant. Suppose by way of contradiction that
$f^{c>a}$ is the constant $1$. This means that $c$ is always ranked on top of $a$. However, since $f^{a>b}$ is non-constant there
exists a value $x$ such that $f^{a>b}(x) = 1$ and similarly there exist a value $y$ such that $f^{b>c}(y) = 1$.
Let $\sigma$ be a ranking whose $a>b$ preference is given by $x$ and whose $b>c$ preferences are given by $y$.
Then $G(\sigma)$ satisfies that $a$ is preferred to $b$ and $b$ is preferred to $c$. Thus by transitivity it follows that
$a$ is preferred to $c$ - a contradiction. The same argument applied if $f^{c>a}$ is the constant $-1$ or if $f^{b>c}$ is a constant function.

We have thus established that all three functions $f^{a>b},f^{b>c}$ and $f^{c>a}$ are of the form $f(x)=x$ of $f(x)=-x$.
To conclude we want to show that all three functions are identical. Suppose otherwise. Then two of the functions have the same sign while the third has a different sign. Without loss of generality assume $f^{a>b}(x) = f^{b>c}(x) = x$ and $f^{c>a}(x)=-x$. Then looking at the profile $a>b>c$ we see that $\sigma' = F(\sigma)$ must satisfy $a>b$ and $b>c$ but also $c>a$ a contradiction. A similar proof applies when $f^{a>b}(x) = f^{b>c}(x) = -x$ and $f^{c>a}(x)=x$.
\end{proof}

\begin{Theorem} \label{thm:3}
Any constitution on three alternatives which satisfies Transitivity, IIA and ND is a dictator.
\end{Theorem}

\begin{proof}
There are two cases to consider. The first case is where two of the functions $f^{a>b},f^{b>c}$ and $f^{c>a}$ are constant.
Without loss of generality assume that $f^{a>b}$ and $f^{b>c}$ are constant. Note that if $f^{a>b}$ is the constant $1$ and
$f^{b>c}$ is the constant $-1$ then $b$ is ranked at the bottom for all social outcomes in contradiction to the ND condition.
A similar contradiction is derived if $f^{a>b}$ is the constant $-1$ and $f^{b>c}$ is the constant $1$.
We thus conclude that $f^{a>b} = f^{b>c}$. However by transitivity this implies that $f^{c>a}$ is also a constant function and
$f^{c>a} = -f^{a>b}$.

The second case to consider is where at least two of the functions $f^{a>b},f^{b>c}$ and $f^{c>a}$ are not constant.
Assume without loss of generality that $f^{a>b}, f^{b>c}$ are non-constant. Therefore, each has at least one pivotal voter.
From Theorem~\ref{thm:tech} it follows that there exists a single voter $i$ such that each of the functions is either constant, or has a single pivotal voter $i$.
We thus conclude that $F$ is of the form $F(\sigma) = G(\sigma(i))$ for some function $G$. Applying Proposition~\ref{prop:1} shows that
either $G(\sigma)=\sigma$ or $G(\sigma)=-\sigma$ and concludes the proof.
\end{proof}

\subsection{ General Proof}
We now prove Theorem~\ref{thm:general}.
\begin{proof}
Note that for any set of three alternatives $A = \{a,b,c\}$ the condition of Theorem~\ref{thm:3} hold for $F_A$. This implies in particular that for all $a,b$ either $f^{a>b}(x)=-x_i$ or $f^{a>b}(x)=x_i$ for some voter $i$.
It remains to show that for all $a,b,c,d$ it holds that $f^{a>b}(x)=f^{c>d}(x)$ which implies that there exist an $i$ such that
$F(\sigma)=\sigma_i$ or $F(\sigma)=-\sigma_i$ as needed.

We first consider the case where $\{a,b\}$ and $\{c,d\}$ intersect in one element, say $d=a$. In this case, Theorem~\ref{thm:3} applied to the rankings of $a,b$ and $c$, implies the desired result.



We finally need to consider the case where $\{a,b\}$ and $\{c,d\}$ are disjoint. Applying Theorem~\ref{thm:3} to $\{a,b,c\}$ we conclude that one of the functions $f^{a>b}(x)=f^{b>c}(x)$. Then applying it to alternatives $\{b,c,d\}$ we conclude that $f^{b>c}(x)=f^{c>d}(x)$. The proof follows.

\end{proof}

\section{The Characterization Theorem}
We now prove Theorem~\ref{thm:very_general}. Given a set of alternatives $A' \subset A$ and an alternative
$b \notin A$, we write $b \sim A'$ if if there exist two alternatives $a,a'\in A_s$ and two profiles $\sigma$ and $\sigma'$ s.t. $F(\sigma)$ ranks $b$ above $a$ and $F(\sigma')$ ranks $a'$ above $b$. Note that if it does not hold that $b \sim A'$ then either $\{b\} >_F A'$ or $A' >_F \{b\}$.

We will use the following lemmas.
\begin{Lemma} \label{lem:gen_ind}
Let $F$ be a transitive constitution satisfying IIA and $A_1,\ldots,A_r,\{b\}$ disjoint sets of alternatives satisfying $A_1 >_F A_2 >_F \ldots >_F A_r$. Then either
\begin{itemize}
\item
There exists an $1 \leq s \leq r+1$ such
\begin{equation} \label{eq:simple_ind}
A_1 >_F \ldots >_F A_{s-1} > \{ b \} >_F A_s >_F \ldots >_F A_r,
\end{equation}
or
\item
There exist an $1 \leq s \leq r$ such that $b \sim A_r$ and
\begin{equation} \label{eq:complicated_ind}
A_1 >_F \ldots >_F A_{s} \cup \{ b \} >_F A_{s+1} >_F \ldots >_F A_r.
\end{equation}
\end{itemize}
\end{Lemma}

\begin{proof}
Consider first the case where for all $s$ it does not hold that $b \sim A_s$.
In this case for all $s$ either $b >_F A_s$ or $A_s >_F b$. Since $b >_F A_s$ implies
$b >_F A_{s+1} >_F \ldots$ and $A_{s'} >_F b$ implies $\ldots >_F A_{s'-1} >_F A_{s'} >_F b$ for all
$s,s'$ by transitivity, equation~(\ref{eq:simple_ind}) follows.

Next assume $b \sim A_s$. We argue that in this case
\[
\ldots >_F A_{s-1} >_F \{b\} >_F A_{s+1} >_F \ldots,
\]
which implies~\ref{eq:complicated_ind}.

Suppose by contradiction that $b >_F A_{s+1}$ does not hold. Then there exists an element $a \in A_{s+1}$
and a profile $\sigma$ where $F(\sigma)$ ranks $a$ above $b$. From the fact that $b \sim A_s$ it follows
that there exist $c \in A_s$ and a profile $\sigma'$ where $F(\sigma')$ ranks $b$ above $c$ above $a$.
We now look at the constitution $F$ restricted to $B = \{a,b,c\}$. For each of $a,b,c$ there exist at least one profile where they are not at the top/bottom of the social outcome. It therefore follows that Theorem~\ref{thm:3} applies
to $F_B$ and that $F_B$ is a dictator. However, the assumption that $A_s >_F A_{s+1}$ implies that $c >_F a$.
A contradiction. The proof that $F_{s-1} >_F b$ is identical.
\end{proof}

\begin{Lemma} \label{lem:gen3}
Let $F$ be a constitution satisfying transitivity and IIA. Let $A$ be a set of alternatives such that
$F_A$ is a dictator and $b \sim A$. Then $F_{A \cup \{b\}}$ is a dictator.
\end{Lemma}

\begin{proof}
Assume without loss of generality that $F_A(\sigma) = \sigma(i)$. Let $a \in A$ be such that there exist a profile
where $F$ ranks $a$ above $b$ and $c \in A$ be such there exists a profile where $a$ is ranked below $c$. Let
$B = \{a,b,c\}$. Then $F_B$ satisfies the condition of Theorem~\ref{thm:3} and is therefore dictator.
Moreover since the $f^{a>c}(x)=x(i)$ it follows that $f^{a>b}(x)=x(i)$ and $f^{b>c}(x)=x(i)$.
Let $d$ be any other alternative in $A$. Let $B = \{a,b,d\}$. Then since $f^{a>b}(x) = f^{a>d}(x) = x(i)$, the conditions of Theorem~\ref{thm:3} hold for $F_B$ and therefore $f^{b>d}(x) = x(i)$. We have thus concluded that
$F_{A \cup\{b\}}(\sigma) = \sigma$ for all $\sigma$ as needed. The proof for the case where $F_A(\sigma) = -\sigma$
is identical.
\end{proof}

Theorem~\ref{thm:3} also immediately implies the following:
\begin{Lemma} \label{lem:gen2}
Let $F$ be a constitution satisfying transitivity and IIA. Let $A$ be a set of two alternatives such that
$F_A$ is not constant and $b \sim A$. Then $F_{A \cup \{b\}}$ is a dictator.
\end{Lemma}

We can now prove Theorem~\ref{thm:very_general}.
\begin{proof}
The proof is by induction on the number of alternatives $k$. The case $k=2$ is trivial. Either
$F$ always ranks $a$ above $b$ in which case $\{a\} >_F \{b\}$ as needed or $F$ is a non-constant function
in which case the set $A =\{a,b\}$ satisfies the desired conclusion.

For the induction step assume the theorem holds for $k$ alternatives and let $F$ be a constitution on $k+1$ alternatives which satisfies IIA and Transitivity. Let $B$ be a subset of $k$ of the alternatives and $b = A \setminus B$.

By by the induction hypothesis applied to $F_B$, we may write $B$ as a disjoint union of $A_1,\ldots,A_r$
such that $A_1 >_F A_2 > \ldots >_F A_r$ and such that if $A_s$ is of size $3$ or more then $F_{A_s}$ is a dictator and if $F_{A_s}$ is of size two then $F_{A_s}$ is non constant. We now apply Lemma~\ref{lem:gen_ind}. If~(\ref{eq:simple_ind}) holds then the proof follows. If~(\ref{eq:complicated_ind}) holds then the proof would follow once we show that $F_C$ is of the desired form where $C = A_s \cup \{b\}$. If $A_s$ is of size $1$ then from the definition of $\sim$ it follows that $F_{A_s \cup \{b\}}$ is non-constant as needed. If $A_s$ is of size $2$ then Lemma~(\ref{lem:gen2}) implies that $F_{A_s \cup \{b\}}$ is a dictator as needed and for the case of $A_s$ of size $3$ or more this follows from Lemma~(\ref{lem:gen3}). The proof follows.
\end{proof}

\section{Conclusion and Followup Work}
A major motivation for the work presented here comes from the desire to obtain quantitative versions of Arrow theorem, where the goal is to derive lower bounds on the probability of a paradox.
Consider $n$ voters who vote independently at random, each following the uniform distribution over the $6$ rankings of $3$ alternatives.
Arrow's theorem implies that any constitution which satisfies IIA and Unanimity and is not a dictator has a probability of at least $6^{-n}$ for a non-transitive outcome, where $n$ is the number of voters. When $n$ is large, $6^{-n}$ is a very small probability, and the question arises if for large number of voters it is possible to avoid paradoxes with probability (exponentially) close to $1$.

In a follow up paper~\cite{Mossel:09b} which uses some of the techniques developed here we derive such a quantitative theorem. In other words, we show that for all $\eps > 0$, there exists $\delta(\eps) > 0$, which does not depend on $n$ such that
if a constitution $F$ satisfies IIA and the probability that $F(\sigma)$ is transitive is at least $1-\delta$, then it is $\eps$ close to a function whose probability of paradox is $0$, i.e., one of the functions listed
in Theorem~\ref{thm:very_general}.

A key feature of this new result is that it {\em does not} use a quantitative version of unanimity.
Indeed,  while it is easy to quantify Non Degeneracy, say by looking at $\min(\P[f=1],\P[f=-1])$, there is no natural way to quantify Unanimity.
 Prior quantitative work on Arrow theorem~\cite{Kalai:02,Keller:09} typically made strong assumptions on $\min(\P[f=1],\P[f=-1])$ and does not allow to obtain quantitative estimates when the number above is smaller than some constant (say $1/100$).

 In addition to the proof ideas presented here, the quantitative proof of Arrow theorem, uses a number of sophisticated mathematical tools including the theory of influences, inverse hyper-contractive estimates and non-linear invariance principles.


\bibliographystyle{abbrv}
\bibliography{all,my}
\end{document}